\documentclass[letterpaper, 10 pt, conference]{ieeeconf}  

\IEEEoverridecommandlockouts                              

\usepackage{amssymb,amsfonts,amsmath,amsthm,amscd}
\usepackage{epsfig,graphics,psfrag}

\newtheorem{thm}{Theorem}

\newtheorem{lemma}{Lemma}

\newtheorem{definition}{Definition}

\def\prob{\mathbb P}
\def\Prob{\mathbb P}
\def\E{\mathbb E}

\def\|{\big|\big|}

\def\u0t{{\tt \underline{0}}}
\def\0t{{\tt 0}}
\def\1t{{\tt 1}}
\def\H{{\mathbb H}}

\def\ind{{\mathbb I}}

\def\reals{{\mathbb R}}

\def\di{{\partial i}}

\def\dmi{{\partial_- i}}

\def\naturals{{\mathbb N}}
\def\dpl{{\partial_+}}
\def\dmn{{\partial_-}}

\def\T{{\sf T}}
\def\F{{\sf F}}
\def\Fh{\widehat{\sf F}}
\def\Z{{\mathbb Z}}
\def\block{{\rm P}_{{\rm \footnotesize err}}}

\def\ve{\varepsilon}
\def\eps{\epsilon}

\def\Typ{{\sf T}}
\def\xh{\widehat{x}}
\def\zh{\widehat{z}}
\def\yh{\widehat{y}}
\def\bm{m_*}
\def\bp{\overline{p}}

\def\vl{\vec{\lambda}}
\def\vW{\vec{W}}
\def\vn{\vec{n}}
\def\T{{\sf T}}

\newcommand{\comment}[1]{}

\title{\Large \bf Detailed Network Measurements Using Sparse Graph Counters:\\
 The Theory}

\author{Yi Lu, Andrea Montanari and Balaji Prabhakar
\thanks{Yi Lu is with the Department of Electrical Engineering,
Stanford University, {\tt\small yi.lu@stanford.edu}. Andrea
Montanari is with Departments of Electrical Engineering and
Statistics, Stanford University, {\tt\small
montanari@stanford.edu}. Balaji Prabhakar is with the Department
of Electrical Engineering, Stanford University, {\tt\small
balaji@stanford.edu}.} }

\begin{document}

\maketitle

\thispagestyle{empty}
\pagestyle{empty}

\begin{abstract}

Measuring network flow sizes is important for tasks like accounting/billing,
network forensics and security.  Per-flow accounting is considered hard
because it requires that many counters be updated at a very high speed; however,
the large fast memories needed for storing the counters are prohibitively
expensive.  Therefore, current approaches aim to obtain approximate flow counts;
that is, to detect large {\em elephant} flows and then measure their sizes.

Recently the authors and their collaborators have developed \cite{our-tech-report}
a novel method for per-flow traffic measurement that is fast, highly memory efficient
and accurate.   At the core of this method is a novel counter architecture called
``counter braids.''  In this paper, we analyze the performance of the counter
braid architecture under a Maximum Likelihood (ML) flow size estimation algorithm
and show that it is optimal; that is, the number of bits needed to store the
size of a flow matches the entropy lower bound.  While the ML algorithm is
optimal, it is too complex to implement.  In \cite{our-tech-report} we have
developed an easy-to-implement and efficient message passing algorithm
for estimating flow sizes.

\end{abstract}
%
%
\section{Introduction}

This paper addresses a theoretical problem arising in a novel approach to
network traffic measurement the authors and their collaborators have
recently developed.  We refer the reader to \cite{our-tech-report} for
technological background, motivation, related literature and other details.
In order to keep this paper self-contained, we summarize the background and
restrict the literature survey to what is relevant for the results of this
paper.

\noindent{\bf Background.}  Measuring the sizes of network flows on high speed
links is known to be a technologically challenging problem
\cite{Varghese03}.  The nature of the
data to be measured is as follows:  At any given time several 10s or 100s of
thousands of flows can be active on core Internet links.  Packets arrive at
the rate of one in every 40-50 nanoseconds on these links which currently run
at 10 Gbps.  Finally, flow size distributions are heavy-tailed, giving rise to
the well-known decomposition of flows into a large number of short ``mice'' and
a few large ``elephants.''  As a rule of thumb, network traffic follows an ``80-20 rule'':
80\% of the flows are small, and the remaining 20\% of the large flows bring about
80\% of the packets or bytes.

This implies that measuring flow sizes accurately requires a large array of counters
which can be updated at very high speeds, and a good counter management algorithm
for updating counts, installing new counters when flows initiate and uninstalling
them when flows terminate.

Since high-speed large memories are either too expensive or simply infeasible with
the current technology, the bulk of research on traffic measurement has focused
on approximate counting methods.  These approaches aim at detecting elephant flows and
measuring their sizes.

\noindent{\bf Counter braids.}  In \cite{our-tech-report} we develop
a novel counter architecture, called ``counter braids'', which is fast, very efficient
with memory usage and gives an accurate measurement of {\em all} flow sizes, not just
the elephants.  We will briefly review this architecture using the following simple
example.

Suppose we are given 5 numbers and are told that four of them are no more than
2 bits long while the fifth can be 8 bits long.  We are not told which is which!

Figures \ref{fig:SimpleCounters} and \ref{fig:CounterBraid} present two approaches
for storing the values of the 5 numbers.  The first one corresponds to a traditional
array of counters, whereby the same number of memory registers is allocated
to each measured variable (flow). The structure in Fig.~\ref{fig:CounterBraid}
is more efficient in memory, but retrieving the count values is less straightforward,
requiring a flow size estimation algorithm.

\begin{figure}[h]
\center{\includegraphics[width=5.6cm]{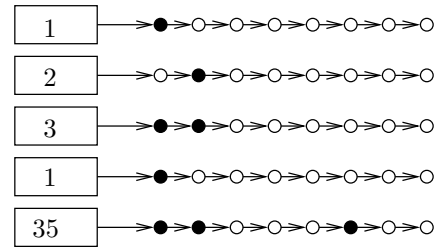}}
\put(-148,80){$1$}
\put(-148,61.5){$2$}
\put(-148,42){$3$}
\put(-148,23.5){$1$}
\put(-152,4){$35$}
\caption{{\small A simple counter structure: to each flow size we associate its
binary representation (filled circle $=1$, empty circle $=0$).}}
\label{fig:SimpleCounters}
\end{figure}
\begin{figure}[h]
\center{\includegraphics[width=5.9cm]{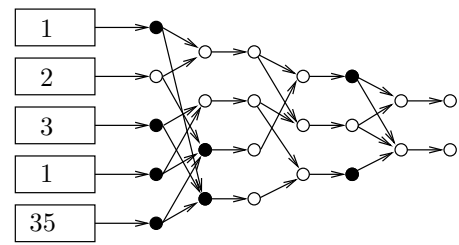}}
\put(-158,78){$1$}
\put(-158,59.5){$2$}
\put(-158,41){$3$}
\put(-158,23){$1$}
\put(-162,4.5){$35$}
\caption{{\small Counters braid.}}\label{fig:CounterBraid}
\end{figure}

Viewed from an information-theoretic perspective, the design of an
efficient counting scheme and a good flow size estimation is equivalent
to the design of an efficient {\em source code} \cite{CoverThomas}.
However, the applications we consider impose a
stringent  constraint on such a code: each time the size of a flow
changes (because a new packet arrives) a small number of operations must
be sufficient to update the stored information.
This is not the case with standard source codes, where changing a single letter
in the source  stream may alter completely the compressed version.

In this paper we prove that, under a probabilistic model for the
flow sizes (namely that they form a vector of iid random variables),
counter braids achieve a compression rate equal to the entropy of the flow
sizes distribution,
in the large system limit. Namely, for any rate larger than
the flow entropy, the flow sizes can be recovered
from the counter values, with error probability vanishing in the large system
limit.
Further, we prove optimal compression can be achieved by
using braids that are \emph{sparse}.
The result is  non-obvious, since counter braids form a
pretty restrictive family of architectures.

Our treatment makes use of techniques from the theory of low-density
parity check codes, and the whole construction is inspired by that of LDPC
\cite{GallagerThesis,MCT}.
The construction of LDPC codes has an analogy in the source coding problem thanks to standard
equivalence between coding over discrete memoryless
symmetric channels, and compressing iid discrete random variables
\cite{CaireEtAl}.
However, the key ideas in the present paper have been developed to deal with
the problem that the flow sizes are \emph{a priori} unbounded.
In the channel coding language, this would be equivalent to use
a countable but infinite input alphabet.

Finally, we insist on using sparse braids for two reasons.
First, this allows the stored values to be updated with a \emph{small}
(typically bounded) number of operations.
Second, it is easy to realize that ML decoding of counter braids
is NP-hard, since it has ML decoding of linear codes as a special case
\cite{Berlekamp}.
However, thanks to the sparseness of the underlying graph,
one may use iterative message passing techniques \cite{Factor}.
Indeed, a simple message passing algorithm for estimating flow sizes is
described and analyzed using real and synthetic network traces in \cite{our-tech-report}.
%
%
\section{Counter Braids: Basic definitions}

\begin{definition}
A \emph{counter braid} is a couple $(G,q)$ where $q\ge 2$ is an
integer (register capacity) and $G$ is a directed acyclic graph on
vertex sets $I$ (input nodes) and $R$ (registers), with the input
nodes having in-degree zero. We  write $G=(I,R,E)$, with $E$ the
set of directed edges in $G$.

For any node $i\in I\cup R$, we will denote by $\dpl i\equiv
\{j:\; (i,j)\in E\}$ the set of descendants of $i$, and by $\dmn
i\equiv \{j:\; (j,i)\in E\}$ the set of parents of $i$. Finally
$\di\equiv \dpl\cup\dmi$.
\end{definition}
In the following we shall often omit the explicit reference to the
register capacity and write $G$ for $(G,q)$. The input size of the
braid is $|I|\equiv n$, and its storage size $|R|\equiv m$. An
important parameter is its rate, which we measure in bits
\begin{eqnarray}
r= \frac{|R|\log_2q}{|I|}\, .
\end{eqnarray}
We will say that a sequence of counters braids $\{G_n = (I_n,R_n,E_n)\}$
is $\emph{sparse}$ if the number of edges per input node
$|E_n|/|I_n|$ is bounded.

\begin{definition}
A \emph{state} (or \emph{configuration} of the counter braid $G_q$, with is
an assignment $(x,y)$ of non-negative integers to the nodes in
$G$, with $x = \{x_i:\, i\in I\}\in\naturals^I$, and
$y=\{y_j:\, j\in R\}\in\naturals^R$. The state $(x,y)$
is \emph{valid} if $y_j\in\{0,\dots,q-1\}$ for any
register $j\in  R$.
\end{definition}
Notice that a valid register configuration can be regarded as an
element of $(\Z_q)^R$ (where $\Z_q$ is the group of integers
modulo $q$.) We denote by $0$ the zero vector in $\naturals^{K}$.

We want now to describe the braid behavior when one
of the input nodes is incremented by one unity (i.e. when
a packet arrives at input node $i\in I$.)
Assume the braid $(G,q)$ to be in a valid state $(x,y)$.
Given $i\in I$, we define the new state $(x',y') = \T_i(x,y)$
by letting $x'_i = x_i+1$, $x'_j=x_j$ for any $j\neq i$,
and $y'$ be defined by the following procedure.
\begin{table}[h]
{\normalsize
\begin{tabular}{ll}
\hline
\multicolumn{2}{l}{REGISTERS UPDATE ({\sc input}: flow index $i$)}\\
\hline
&\vspace{-0.2cm}\\

1: & $y_j(0)=y_j$ for $j\not\in \dpl i$,\\
   & \hspace{0.1cm} and $y_j(0)=y_j+1$
otherwise.\\
2: & Set $t=0$.\\
3: & {\bf while} $y(t)$ is not valid\\
4: & \hspace{0.5cm}Let $j\in R$ be such that $y_j(t)\ge q$;\\
5: & \hspace{0.5cm}Set $y_j(t+1)=y_j(t)-q$;\\
6: & \hspace{0.5cm}For any $l\in\dpl j$, set $y_l(t+1)=y_l(t)+1$;\\
7: & \hspace{0.5cm}For any $l\in R\setminus \{j,\dpl j\}$, set
$y_l(t+1)=y_l(t)$;\\
8: & \hspace{0.5cm}Increment $t:=t+1$;\\
9:& {\bf end}\\
10:& {\bf return} $y(t)$.\\
\hline
\end{tabular}}
\end{table}
Notice that this definition is ambiguous in that
we did not specify which register to pick among the ones with $y_j(t)\ge q$
at step 4 in the registers update routine. However this is not necessary, as
stated in the following lemma (the proof is omitted from this extended abstract).
\begin{lemma}
The update procedure above halts after a finite number of steps.
Further its output $\T_i(x,y)$ does not depend on
the order of update of the registers.
\end{lemma}
With an abuse of notation we shall write $x' = \T_i(x)$, $y' = \T_i(y)$,
when $(x',y') = \T_i(x,y)$.

When input values $x$ are incremented sequentially, the stored
information $y$ is updated according to the above procedure. From
now on we shall take a static view and assume a certain input $x$.
The corresponding stored information $y$ is obtained through the
mapping defined below.
\begin{definition}
Given a counter braid $(G,q)$, the associated \emph{storage
function} $\F_G:\naturals^I\to \Z_q^R$ returns, for any input
configuration $x\in \naturals^I$ a register configuration
$y=\F_G(x)\in \Z_q^R$ defined as follows. Let $x^{(0)}=0$,
$x^{(1)},\dots$, $x^{(N)}=x$ be a sequence of input configurations
such that $x^{(s+1)}$ is obtained from $x^{(s)}$ by incrementing
its entry $i(s)$. Then
\begin{eqnarray}
\F_G(x) \equiv \T_{i(N)}\circ\T_{i(N-1)}\circ\cdots\circ \T_{i(1)}(0)\, .
\end{eqnarray}
\end{definition}
We shall drop the subscript $G$ from $\F_G$ whenever clear from the context.
A priori it is not obvious that the mapping $\F_G$ is well defined.
In particular, it is not obvious that it does not depend on the order
in which input values are incremented, i.e. on the sequence
$\{i(1),\dots,i(N)\}$. This is nevertheless the case
(the proof is omitted.)

\begin{definition}
Given a counter braid $(G,q)$, a \emph{reconstruction} (or \emph{decoding})
\emph{function} is a function $\Fh :\Z_q^R\to\naturals$.
\end{definition}
%
%
\subsection{Main results}

Throughout this paper, we shall model the input
values as iid integer random variables $(X_1,\dots,X_n)\equiv X$
(identifying $V=[n]$) with common distribution $p$.
The (binary) entropy of this distribution will be denoted by
$H_2(p) \equiv -\sum_x p(x)\log_2 p(x)$.
\begin{definition}
A sequence of  counters braids $\{G_n = (I_n,R_n,E_n)\}$, with
$|I_n|=n$ has
\emph{design rate} $r$ if
\begin{eqnarray}
r = \lim_{n\to\infty} \frac{|R_n|}{|I_n|}\log_2 q\, .
\end{eqnarray}
It is \emph{reliable} for the distribution $p$
if there exists a sequence of
reconstruction functions $\Fh_n\equiv \Fh_{G_n}$ such that,
for $X$ a random input and $Y\equiv\F_{G_n}(X)$
\begin{eqnarray}
\block(G_n,\Fh_n) \equiv \prob\{\Fh_n(Y)\neq X\}\stackrel{n}{\to} 0\, .
\end{eqnarray}
\end{definition}

Shannon's source coding theorem implies that there cannot exist
reliable counter braids with asymptotic rate $r<h_2(p)$. However,
the achievability of such rates is far from obvious, since counter
braids are a fairly specific compression scheme. The main theorem
of this paper establishes achievability, even under the
restriction that the braid is sparse.

In order to avoid technical complication, we make two assumptions on the input
distribution $p$:
\begin{enumerate}
\item It has \emph{at most power-law tails}. By this we mean that
$\prob\{X_i\ge x\}\le Ax^{-\epsilon}$ for some $\epsilon>0$.
\item It has \emph{decreasing digit entropy}. Let
$X_i = \sum_{a\ge 0} X_i(a)q^a$ be the $q$-ary expansion of $X_i$,
and $h_l$ be the $q$-ary
entropy of $X_i(l)$. Then $h_l$ is monotonically decreasing in
$l$ for any $q$ large enough.
\end{enumerate}
We call a distribution $p$ with this two properties \emph{admissible}.
While this class does not cover all possible distributions, it is likely
to include any case of practical interest.
\begin{thm}\label{thm:MainThm}
For any admissible input distribution $p$, and any rate
$r>H_2(p)$ there exist a sequence of reliable sparse
counter braids with asymptotic rate  $r$.
\end{thm}
As stressed above, we insist on the braid being sparse for two
reasons: $(i)$ It allows to update the registers content $y$ with a
small number of operations, whenever one entry of $x$ is incremented
(i.e. the storage function can be efficiently recomputed);
$(ii)$ It opens the way to using low-complexity message passing algorithms
for estimating the input vector $x$, given the stored information
(i.e. for evaluating the recovery function $\Fh_G$).
%
%
\section{The architecture}

\subsection{Layering}

We will consider \emph{layered} architectures. By this we mean that
the set of register is the disjoint union of $L$ layers
$R = R^1\cup R^2\cup\dots\cup R^{L}$ and that directed edges
are either from $I$ to $R^1$ or from $R^{l}$ to $R^{l+1}$
for some $l\in\{1,\dots,L-1\}$
(we shall sometimes adopt the convention $R^0\equiv I$).
We denote by $y^{(l)} = \{y_i:\, i\in R^l\}$ the vector of register
values in layer $l$. We further let $m_l\equiv|R^l|$ denote the size
of the $l$-th layer (with $m_0\equiv n$).

The graph structure is conveniently encoded in $L$
matrices $\H_1$, \dots $\H_L$, whereby $\H_l$ is the $m_l\times m_{l-1}$
adjacency matrix of the subgraph induced by $R^{l}\cup R^{l-1}$.
We further let $\H^l = \H_l\cdot \H_{l-1}\cdots \H_1$.
The storage function $\F$ can be characterized as follows.
\begin{lemma}\label{lemma:Expansion}
Consider an $L$-layers counters braid, let $x$ be its input,
and define the sequence of vectors $z^{(l)}\in \naturals^{R^l}$,
by $z^{(0)} = x$ and
\begin{eqnarray}
z^{(l)} = \lfloor (\H_l z^{(l-1)})/q \rfloor\, .\label{eq:ZRecursion}
\end{eqnarray}
(the division and floor operation being component-wise on the vector
$\H_l z^{(l-1)}$.) Then, the register values are
$y^{(l)}= \H_l z^{(l-1)} \mod q$.
\end{lemma}
%
%
\subsection{Recovery function}\label{sec:Recovery}

We now describe the recovery function $\Fh$. Since in this
paper we are only interested in achievability, we will
neglect complexity considerations.

\subsubsection{One layer}
Let us start from a one-layer braid and assume the
inputs to be iid with common distribution
$p_*$ supported on  $\{0,\dots q-1\}\ni x_i$. Then the register values
are $y = \H x$ $\mod q$, where $\H$ is the adjacency matrix of the braid.
Fix $\gamma\in (0,1)$.
We say that the input $x\in\{0,\dots,q-1\}^n$ is \emph{typical},
and write $x\in \Typ_n(p_*)$ if its type
$\theta_x$  satisfies $D(\theta_x|| p_*)\le n^{-\gamma}$
(here the Kullback-Leibler
divergence is computed in natural base).
Denote by $\Typ_n(p_*;y)$ the set of input vectors that are typical and
such that $\H x = y$ $\mod q$. The `typical set decoder'
returns a vector $\xh$ if this is the the unique element in $\Typ_n(p_*;y)$
and a standard error message otherwise. In formulae
\begin{eqnarray}
\Fh(y) =\left\{\begin{array}{ll}
\xh & \mbox{if $\Typ_n(p_*;y)=\{\xh\}$,}\\
\ast & \mbox{if $|\Typ_n(p_*;y)|\neq 1$.}
\end{array}\right.
\end{eqnarray}

\subsubsection{Multi-layer} Consider now a multi-layer braid and
$x\in \naturals^{I}$ (inputs not restricted to be smaller than $q$)
with $x_i$'s distributed independently according to $p$.
It is convenient to write the input vector in base $q$
\begin{eqnarray}
x = \sum_{a\ge 0} x(a)\, q^a\, . \label{eq:XExpansion}
\end{eqnarray}
where $x(a) = \{x_i(a):\, i\in V\}$ with $x_i(a)\in\{0,\dots,q-1\}$.
Notice that, for each $a\ge 0$, the vector $x(a)$ has iid entries.
Let $p_a$ be the distribution on $x_i(a)$ when $x_i$ has distribution
$p$.

We'll apply typical set decoding recursively, determining
the $q$-ary vectors $x(0)$, $x(1)$, $x(2)$, $\dots$ in this order.
Consider first $x(0)$. It is clear from Lemma \ref{lemma:Expansion}
that $y^{(1)}= \H_1 x(0) = \H^1 x(0)$ $\mod q$. We then apply
typical set decoding to the determination of $x(0)$. More precisely,
we look for a solution of $\H^1x = y^{(1)}$ $\mod q$ that is
typical under distribution $p_0$. If there is a unique such solution,
we declare it our estimate of $x(0)$ and denote by $\xh(0)$.
Otherwise we declare an error.

Consider now the determination  of $x(l)$ and assume the lower
order terms in the expansion (\ref{eq:XExpansion}) have already
been estimated to be $\xh(0)$, $\xh(1)$, $\dots$, $\xh(l-1)$. Let
$\zh^{(0)} \equiv \sum_{a=0}^{l-1}\xh(a)\, q^{a}$, and
$\zh^{(a)}$, $a\ge 1$ be determined through the same recursion as
in Eq.~(\ref{eq:ZRecursion}). Further let $\yh^{(a)} =
\H_{a}\zh^{(a-1)}\mod q$ (this are nothing but the register values
on input  $\zh^{(0)}$).

Assume the estimates $\xh(0)$, $\xh(1)$, $\dots$, $\xh(l-1)$
to be correct. It is then easy to realize that $\yh^{(a)}= y^{(a)}$
for $a=1,\dots l$. Further $z^{(l)} = \zh^{(l)} + \H^l x(l)$ $\mod q$, hence
\begin{eqnarray}
y^{(l+1)} = \yh^{(l+1)} + \H^{l+1} x(l)\, \mod q\, .
\end{eqnarray}
We therefore proceed to compute $y^{(l+1)}-\yh^{(l+1)}$ $\mod q$.
If the linear system $\H^{l+1} x(l) = y^{(l+1)}-\yh^{(l+1)}$ $\mod
q$ admits more than one or no solution that is typical with
respect to the distribution $p_l$, an error is returned.
Otherwise, the next term in the expansion (\ref{eq:XExpansion}) is
estimated through the unique typical solution of such linear
system.

The recovery algorithm is summarized below, with one improvement with
respect to the description above. Instead of recomputing
$\zh^{(0)}$, $\dots, \zh^{(l)}$, at stage $l$ we only compute the
vector $\zh^{(l)}$ that is needed at the present stage.
\begin{table}[h]
{\normalsize
\begin{tabular}{ll}
\hline
\multicolumn{2}{l}{RECOVERY ({\sc input}: register values $y$)}\\
\hline
&\vspace{-0.2cm}\\
1: & Initialize $\zh^{(a)}=0$ for $a\ge 0$;\\
2: & {\bf for $l\in\{0,\dots L\}$}\\
3: & \hspace{0.5cm}Set $\yh(l+1)= \H_{l+1}\zh^{(l)}$ $\mod q$;\\
4: & \hspace{0.5cm}Let $\Typ_l$ be the set of $p_l$-typical\\
   & \hspace{0.75cm}solutions of $\H^{l+1}
\xh = y^{(l+1)}-\yh^{(l+1)}$, $\mod q$;\\
5: & \hspace{0.5cm}If $\Typ_l=\{ \xh\}$ let $\xh(l) = \xh$\\
   & \hspace{0.75cm}otherwise if $|\Typ_l|\neq 1$ return error;\\
6: & \hspace{0.5cm}Set $\zh^{(l+1)} =
\lfloor\{\H^{l+1}\zh^{(l)}+\H_{l+1}\xh(l)\}/q\rfloor$;\\
7:& {\bf end}\\
8:& {\bf return} $\xh = \sum_i\xh(i)\, q^i$.\\
\hline
\end{tabular}}
\end{table}
%
%
\subsection{Sparse graph ensemble and choice of the parameters}

The optimal compression rate in Theorem \ref{thm:MainThm} is achieved
with the following random sparse graph construction.
Fix the registers capacity $q$ and an integer $k\ge 2$.
Then for $l=1,\dots, L_0$ the graph induced by vertices
$R_{l-1}\cup R_{l}$ has a random edge set that is sampled by
connecting each $i\in R_{l-1}$ to $k$ iid uniformly random vertices in $R_{l}$
(all edges being directed from $R_{l-1}$ to $R_{l}$).
In other words, the $m_l\times m_{l-1}$,
$0-1$ matrix $\H_{l}$ has independent columns,
each sampled by incrementing $k$ iid positions.

The choice of this ensemble is motivated by implementation
concerns. In the flow counting problem, we do not know a priori
the exact number of flows that needs to be stored. The above
structure, this number can be changed without modifying existing
links. Further, for each new flow, the subset of $k$ registers it
is connected to can be chosen through a simple hash function.

To these $L_0$ stages, we add further $L_1$ stages, all of the same size
$m_{L_0+1}=\cdots = m_{L_0+L_1}=\bm$, with edges connecting
each node in $R_{l-1}$ to a different node in $R_l$.
Equivalently, we take $\H_l$ to be the identity matrix in these stages.

It remains to specify the number of stages $L_0$, $L_1$ and their
sizes $m_1$,\dots, $m_{L_0}$. Let $p_l$ be the distribution of the
$l$-th least significant digit in the $q$-ary expansion of $X_i$.
Recall that we defined $h_l$ to be the $q$-ary entropy of the
distribution $p_l$, i.e.
\begin{eqnarray}
h_l \equiv -\sum_{x=0}^{q-1}p_l(x)\log_q p_l(x)\, .
\end{eqnarray}
Finally, in the achievability proof, we shall assume that $q$ is a
prime number, large enough for $h_l$ to be monotonically
decreasing.
\begin{lemma}\label{lemma:Entropy}
Assume $\prob\{X_1\ge x\}\le A\, x^{-\epsilon}$. Then there exists
constants $B, C$ that only depend on $A,\epsilon$, such that
for all $l\ge 1$, and all $q$ large enough
\begin{eqnarray}
&&h_l\le B\, l\, q^{-l\epsilon}\; ,\label{eq:DigitEntropy1}\\
&&\Big|h_2(p)-\sum_{l\ge 0} h_l\log_2 q\Big|\le C\, q^{-\epsilon}
(\log_2 q)^2\, .\label{eq:DigitEntropy2}
\end{eqnarray}
\end{lemma}
The proof of this simple Lemma is deferred to Section \ref{sec:Lemmas}.

\begin{lemma}\label{lemma:TypicalSize}
Let $p_*$ be a distribution over $\{0,\dots, q\}$, with $q$-ary
entropy $H(p_*)$, and $\Typ_n(p_*)$ be the set of $p_*$-typical
vectors defined as in Sec \ref{sec:Recovery}. Let $|\Typ_n(p_*)|$ be the size of
this set. Recall that $x\in \Typ_n(p_*)$ if its type $\theta_x$
satisfies $D(\theta_x|| p_*)\le n^{-\gamma}$. Then, for any
$\beta\in (1-\gamma/2,1)$, there exists $A = A(\beta,\gamma,q)$
such that
\begin{eqnarray*}
|\Typ_n(p_*)|\le q^{n H(p_*)+A\, n^{\beta}}\, ,
\, .\label{eq:TypeCount}
\end{eqnarray*}
Further, if $X =(X_1,\dots,X_n)$ is a vector with iid entries
with common distribution $p_*$
\begin{eqnarray}
\prob\left\{X\not\in\Typ_n(p_*)\right\}\le (n+1)^q\, e^{-n^{1-\gamma}}\, .
\label{eq:Sanov}
\end{eqnarray}
\end{lemma}
In the following we will consider $\gamma$ and $\beta$ fixed once and for all,
for instance by $\gamma=1/2$ and $\beta = 7/8$.

Fix some $\delta>0$, and let $A(q)$ be a suitably large constant,
we let, for $l=1,\dots, L_0$,
\begin{eqnarray}
m_l & \equiv & \max\{ \underline{m}_l, \lceil \delta\, m_{l-1}\rceil\}\, ,\\
\underline{m}_l &\equiv &
(1+\delta)\left\lceil n h_{l-1}(1+\delta)+A(q)n^{\beta}\right\rceil\, .
\end{eqnarray}
The number of stages is such that
\begin{eqnarray}
\underline{m}_l\le n(\log n)^{-2} \mbox{ for any $l\ge L_0$}\, .
\end{eqnarray}
This implies, by Lemma \ref{lemma:Entropy}, $L_0 = O(\log\log n)$.
To this we add $L_1= (\log n)^{3/2}$ stages within the second group,
of size $\bm=m_{L_0}\le n(\log n)^{-2}$.
The total number of registers is therefore
upper bounded as
$|R|\le n(1+\delta)\sum_{l\ge 0} (h_l+An^{\beta-1}) (\sum_{i\ge 0}\delta^i)
+ n(\log n)^{-1/2}$,
and therefore the asymptotic rate of this architecture
\begin{eqnarray}
r\le \frac{1+\delta}{1-\delta}\sum_{l\ge 0} h_l\log_2 q\, .
\end{eqnarray}

Since the right hand side can be made arbitrarily close
to $h(p)$ by Lemma \ref{lemma:Entropy}, Theorem \ref{thm:MainThm}
follows from the following.
\begin{thm}\label{thm:Rephrasing}
For any input distribution $p$ with at most power-law tails and
any choice of $q\ge 2$ and $\delta>0$, there exists $k\ge 2$ such
that the multi-layer braid described above is reliable.
\end{thm}
%
%
\section{Analysis of one-layer architectures}

In order to prove our main Theorem or, equivalently, Theorem
\ref{thm:Rephrasing}, we need first to prove a few preliminary
results concerning a one-layer architecture. The proof here follows
the technique of \cite{McEliece}, the main tool being
an estimate of the distance enumerator as in
\cite{GallagerThesis,MB04,LS02}. Distance enumerators for non-binary LDPC
codes have been estimated in \cite{Burshtein1}.
Unhappily we cannot here limit ourselves to citing these works, because
the graph ensemble is different from the regular ones treated there.

Throughout this
Section the source is a vector $X=(X_1,\dots,X_n)$ with iid
entries taking values in $\{0,\dots, q-1\}$ and distribution $p_*$
(in the application to multi-layer schemes $p_*$ will coincide
with $p_l$ for some $l\ge 0$). We let $\H$ be an $m\times n$
matrix whose columns are independent vectors with integer entries.
Each column is obtained by choosing $k$ positions independently
and uniformly at random (eventually with repetition) and
incrementing the corresponding entry by one. In other words, $\H$
is distributed as the adjacency matrix of a given layer in the
multi-layer architecture.

Our first result is a simple combinatorial calculation.
Let $\vl = \{\lambda_z:\, z=1,\dots, q-1\}$ be a vector in
$\reals_+^{q-1}$.
It is convenient to introduce the random variable
$\vW = \{W_z:\, z=1,\dots,q-1 \}$ taking values in $\naturals^{q-1}$.
The joint distribution of $(W_1,\dots, W_{q-1})$ (to be denoted by
$\prob_{\vl}$) is the one
of $q-1$ Poisson random variables with means (respectively)
$\lambda_1,\dots,\lambda_{q-1}$, conditioned on
$\sum_{z=1}^{q-1}z W_z = 0$, $\mod q$.
\begin{lemma}\label{lemma:ProbSol}
Let $x\in\{0,\dots, q-1\}^n$ be an input vector with
$n_z$ entries equal to $z$, for $z=0,\dots, q-1$, and $\H$ be a random
matrix as above. Define $\vn = \{n_z:\, z=1,\dots,q-1\}$.
For any $\vl\in\reals_+^{q-1}$, let
$\vW_1$,\dots $\vW_m$ be $m$ iid vectors with distribution $\Prob_{\vl}$.
Then the probability that $\H x =0$ $\mod q$ is
\begin{align}
\prob\{\H x = 0\} = \prod_{z=1}^{q-1}
\frac{(kn_z)!\, e^{m\lambda_z}}{(m\lambda_z)^{kn_z}}\, Q(\vl)^{m}\,
\prob_{\vl}\left\{\sum_{i=1}^m \vW_i = k\vn\right\}\, ,
\label{eq:ExactProb}
\end{align}
where $Q(\vl)$ is the probability that $\sum_{z=1}^{q-1}z\, U_z = 0$, $\mod q$
for independent Poisson random variables with means
$\lambda_z$.

Further, for some universal constant $C$, and $D_q = 1-\cos 2\pi/q$,
and $n_*=\sum_z n_z$
\begin{eqnarray}
\prob\{\H x = 0\} & \le &
(Ck n_*)^{\frac{q-1}{2}}Q(k\vn/m)^mR(m,\frac{k}{q}\sum_{z=1}^{q-1}zn_z)\, \nonumber\\
&&\label{eq:BoundProb}\\
Q(k\vn/m)&\le &  \frac{1}{q} \left[ 1+(q-1)\, e^{-D_q
\frac{kn_*}{m}}  \label{eq:BoundProb2}
\right]\, ,\\
R(m, N) & = & \min\left\{1, (Cq^2N/m)^{N}\right\}
\label{eq:BoundProb3}
\end{eqnarray}
\end{lemma}
\begin{proof}
Due to the symmetry of the distribution of $\H$ with respect to
permutation in its columns, $\prob\{\H x = 0\}$ does depend on $x$
only through the number of ones, twos, etc. Without loss of
generality we can assume the first $n_1$ coordinates to be ones,
the next $n_2$ to be twos, and so on, and neglect the last
$n-\sum_z n_z$ columns, corresponding to zeros. Think now of
filling the matrix, by choosing its non-zero entries (edges in the
associated graph). If we associate to each such entry the value of
the corresponding coordinate in $x$, we want the probability for
the sum of labels on each row to be $0$ $\mod q$. Since entries
are independent and uniformly random, this is equal to the
probability that each of $m$ urns is filled with balls whose
labels add to $0$, when we throw $kn_1$ balls labeled with $1$,
$kn_2$ labeled with $2$, and so on. It is an exercise in
combinatorics to show that this is
\begin{eqnarray*}
\prod_{z=1}^{q-1}\frac{(kn_z)!}{m^{kn_z}}\, {\sf coeff}\left\{
P(\xi_1,\dots,\xi_{q-1})^m,\xi_1^{kn_1}\cdots\xi_{q-1}^{kn_{q-1}}\right\}\, ,\\
P(\cdots)\equiv\sum_{l_1\dots l_{q-1}}\frac{\xi_1^{l_1}}{l_1!}
\dots \frac{\xi_{q-1}^{l_{q-1}}}{l_{q-1}!}\,\ind\left\{\sum_{z=1}^{q-1}
z\, l_z = 0\right\}\, .
\end{eqnarray*}
Equation (\ref{eq:ExactProb}) is then obtained by evaluating
$\Prob_{\vl}$ and showing that it yields the above combinatorial
expression.

In order to get Eq.~(\ref{eq:BoundProb}), we denote
$\prob_{\vl}\{\cdots\}$ by $R$, and use $\lambda_z=k n_z/m$, thus
leading to
\begin{eqnarray*}
\prob\{\H x=0\} = \prod_{z=1}^{q-1}
\frac{(kn_z)!\, e^{kn_z}}{(kn_z)^{kn_z}}\, Q(k\vn/m)^{m} R\, .
\end{eqnarray*}
Equation (\ref{eq:BoundProb}) follows from the observation that
$N!\le \sqrt{CN}\, (N/e)^N$ for some universal constant $C$.

In order to prove Eq.~(\ref{eq:BoundProb2}), notice that,
by discrete Fourier transform
\begin{eqnarray*}
Q(\vl) & = & \frac{1}{q}\sum_{\ell = 0}^{q-1}
\E\left\{e^{\frac{2\pi i\ell}{q}\sum_{z=1}^{q-1}zU_z }\right\} \\
& = & \frac{1}{q}\sum_{\ell = 0}^{q-1}
\exp\left\{-\sum_{z=1}^{q-1}\lambda_z(1-e^{\frac{2\pi i\ell z}{q}})\right\}
\, .
\end{eqnarray*}
The claim is proved by singling out the $\ell=0$ term and bounding
the others using ${\rm Re}(1-e^{\frac{2\pi i\ell z}{q}})\ge D_q$.

Let us finally prove Eq.~(\ref{eq:BoundProb3}). Obviously $R\le 1$
since it is an upper bound on the probability
$\prob_{\vl}\{\cdots\}$. If we let $N\equiv
\frac{k}{q}\sum_{z=1}^{q-1}zn_z$, we can therefore assume, without
loss of generality, that $N$ is an integer with $N/m\le 1/q$. Let
$V_i$ be distributed as $\sum_{z=1}^{q-1}W_{i,z}z/q$ conditioned
on $V_i$ being an integer. Then the probability
$\prob_{\vl}\{\cdots\}$ is upper bounded by
\begin{eqnarray*}
\prob\left\{\sum_{i=1}^m V_i \ge N\right\}&\le
& \binom{m}{N}\prob\{V_i\ge 1\}^N\\
&\le &\left(\frac{Cm}{N}\, \prob\{V_i\ge 1\}\right)^N\, .
\end{eqnarray*}
Recalling the definition of $V_i$, we have
\begin{eqnarray*}
\prob\{V_i\ge 1\} & = & \prob\Big\{\sum_{z=1}^{q-1}z U_z\ge q\Big|
\sum_{z=1}^{q-1}z U_z = 0\mod q \Big\}\\
&\le &e^{\sum_{z=1}^{q-1}\lambda_z}
\prob\Big\{\sum_{z=1}^{q-1}z U_z\ge q\Big\}\, .
\end{eqnarray*}
But $\sum_{z=1}^{q-1}\lambda_z = kn_*/m\le Nq/m\le 1$. Further,
$\sum_{z=1}^{q-1}z U_z\ge q$ only if $\sum_{z=1}^{q-1} U_z\ge 2$.
Therefore we get
\begin{eqnarray*}
\prob\{V_i\ge 1\}&\le &e \prob\Big\{\sum_{z=1}^{q-1} U_z\ge 2\Big\}\\
&\le & C\Big(\sum_{z=1}^{q-1}\lambda_z\Big)^2\le C(kn_*/m)^2\, .
\end{eqnarray*}
The proof is completed by noticing that $(kn_*/m)\le (Nq/m)$.
\end{proof}

In the following, given a vector $x = (x_1,\dots,x_n)$, we shall denote by
$||x||_0$ is number of non-zero entries.
\begin{lemma}\label{lemma:SmallWeigth}
Let $\H$ be a random $m\times n$ matrix  distributed as above,
with column weight $k$. Assume $k$ not to be a multiple of $q$,
$m\le n$, $(m/nk)^{1/k}\ge \Delta>0$ and $3\le k\le m/\log m$.
Then, there exists a constant $B=B(q,\Delta)$, $C = C(q,\Delta)$,
such that, if
\begin{eqnarray}
E\le \frac{C\, m}{\log(nk/m)}\, ,\label{eq:LHypothesis}
\end{eqnarray}
then
\begin{eqnarray}
\prob \Big\{\exists ||z||_0\le E\, :\,\H z = 0 \mod q\Big\}\le
 n^{2}\left(\frac{Bk}{m}\right)^{\frac{k}{q}}\label{eq:BoundSmallDist}
\end{eqnarray}
(where it is understood that $z\in\{0,\dots,q-1\}^{n}$.)
\end{lemma}
\begin{proof}
Throughout the proof, $A$ will denote a generic
constant depending only on $q$ that can be chosen large enough to make the
inequalities below hold.

Let $z\in\{0,\dots,q-1\}^n$ be such that $||z||_{0}=\ell$.
We will upper bound the probability that $\H z = 0$ $\mod q$
in different ways depending whether $\ell \le E_0$ or $\ell>E_0$,
where
\begin{eqnarray}
E_0 = \rho(q)\, \frac{m}{k}\,\left(\frac{m}{nk}\right)^{2/(k-2)}\, ,
\label{eq:L0Definition}
\end{eqnarray}
with $\rho(q)$ a function to be determined. Notice that, under our
hypotheses,
\begin{eqnarray}
\frac{k E_0}{m}\ge \rho(q)\, \Delta^{2k/(k-2)}\, \label{eq:L0Bound}
\end{eqnarray}
is bounded away from $0$ (as $2< 2k/(k-2)\le 6$
for $k\ge 3$.)

For $||z||_0=\ell\le E_0$ (and $z\neq 0$) we use Lemma
\ref{lemma:ProbSol}, Eq.~(\ref{eq:BoundProb}), where we set $Q(\,
\cdots\,)\le 1$, $n_*=\ell$ and $\frac{k\ell}{q}\le
\frac{k}{q}\sum_{z=1}^{q-1}zn_z\le k\ell$. Further we assumed
$Ak\ell/m\le 1$, which holds without loss of generality if we take
$\rho(q)\le 1/A\Delta^{6}\le 1/A\Delta^{2k/(k-2)}$ in
Eq.~(\ref{eq:L0Definition}), thus getting
\begin{eqnarray}
\prob\{\H z = 0\} & \le &
(Ak \ell)^{\frac{q-1}{2}} (Ak\ell/m)^{k\ell/q}\, .
\end{eqnarray}

Since $(k\ell)^{(q-1)/2}\le A^{k\ell/q}$, we have (by properly
adjusting $A$)
\begin{eqnarray}
\prob\{\H z = 0\} & \le & (Ak\ell/m)^{k\ell/q}\, .
\end{eqnarray}

For $||z||>E_0$, we use Eq.~(\ref{eq:BoundProb}) with $R(\cdots)\le 1$.
Since $k\ell/m > k E_0/m$ is bounded away from $0$ by Eq.~(\ref{eq:L0Bound}),
we have $Q(\cdots)\le e^{-C}$ for some $C=C(\Delta,q)>0$ and therefore
\begin{eqnarray}
\prob\{\H z = 0\} & \le & (Ak\ell)^{(q-1)/2}\, e^{-Cm}\, .
\end{eqnarray}

There are at most $\binom{n}{\ell}\, (q-1)^{\ell}\le
\left(\frac{A n}{\ell}\right)^\ell$  vectors $z$ with $||z||_0 = \ell$.
If we denote by $\prob_{E_1,E_2}$ the probability of the event
$\Big\{\exists z: \, E_1\le ||z||_0\le E_2\, ,
\,\H z = 0 \mod q\Big\}$, the probability in Eq.~(\ref{eq:BoundSmallDist})
is upper bounded by $\prob_{2,E_0}+\prob_{E_0,E}$
(notice that if $k$ is not a multiple of $q$,
$\H z = 0$ is impossible for $||z||_0=1$).
By union bound we have
\begin{eqnarray*}
\prob_{2,E_0}&\le & \sum_{\ell=2}^{E_0} \left(\frac{A n}{\ell}\right)^\ell
 \left(\frac{Ak\ell}{m}\right)^{k\ell/q}
\\
&\le& \left(\frac{A n}{2}\right)^2
 \left(\frac{2Ak}{m}\right)^{2k/q}
\sum_{\ell=2}^{E_0} \xi(\ell)^{\ell-2} \, ,
\end{eqnarray*}
where (using $(\ell/2)^{2k/q-2}\le A^{k(\ell-2)/q}$ and
eventually adjusting the constant $A$)
\begin{eqnarray*}
\xi(\ell) \equiv \frac{n}{\ell}
 \left(\frac{Ak\ell}{m}\right)^{k/q}\, .
\end{eqnarray*}
For $\ell\le E_0$, and choosing $\rho(q)$
small enough in Eq.~(\ref{eq:L0Definition}), we obtain $\xi(\ell)\le 1/2$
thus leading to  $\prob_{2,E_0}\le n^2 (Ak/m)^{k/q}$.

Finally consider the contribution of vectors $||z||_0\ge E_0$.
Proceeding as above, we have
\begin{eqnarray*}
\prob_{E_0,E}&\le &\sum_{\ell=E_0}^{E} \left(\frac{A n}{\ell}\right)^\ell
(Ak\ell)^{(q-1)/2}\, e^{-Cm}\\
&\le & E \left(\frac{A n}{E}\right)^E (AkE)^{(q-1)/2}\, e^{-Cm}\, .
\end{eqnarray*}
Here we bounded $(An/\ell)^{\ell}=[ (An/\ell)^{\ell/An}]^{An}\le
(An/E)^{E}$, using the fact that $x^{-x}$ is an increasing
function of $x$ for $x\le e^{-1}$, and that $E/An =
Cm/An\log(nk/m)\le Cm/An$ is smaller than $e^{-1}$ for $C$ small
enough.

Finally we bound $E^{(q+1)/2}\le A^E$ and $k^{(q-1)/2}\le k^{E}$
(which holds for $m$ large enough), thus getting
\begin{eqnarray*}
\prob_{E_0,E}&\le & \left(\frac{nkA}{m}\right)^E\, e^{-Cm}\, .
\end{eqnarray*}
If we take $E= Cm/2\log(nkA/m)$, we get $\prob_{E_0,E_1}\le
e^{-Cm/2}$, which is smaller than $(Bk/m)^{\frac{k}{q}}$ for a
properly chosen constant $B$ and $k\le m/\log m$ (indeed $k\le
m\ve_m$ would be enough for any $\ve_m\downarrow 0$.)
\end{proof}
%

%
%
\section{Analysis of multi-layer architectures and
proof of Theorem \ref{thm:MainThm}}

\begin{proof}
Let $\block^{(l)}$ denote the probability that $l$-th term in the
$q$-ary expansion of $x$ is decoded incorrectly by the decoder in
Section \ref{sec:Recovery} (i.e. that $\xh(l)\neq x(l)$) given
that $x(0)$, \dots, $x(l-1)$ have been correctly recovered. We
will prove that $\block^{(l)} = O(n^{-A})$ for some $A>0$. Since
the multi-layer architecture involves at most  $C(\log
n)^{\frac{3}{2}}$ layers, this implies the thesis. Further, we
shall consider only the first $L_0$ layers, since it will be clear
from the derivation below that the error probability is decreasing
for the last $L_1$ layers.

Let $x$ be the input. Since we are focusing on the $l$-th term in
the $q$-ary expansion  of the input, we will drop the index $l$,
and take $x\in \{0,\dots,q-1\}^n$. This is just a vector whose
entries are iid with distribution $p_l$.

The error probability $\block^{(l)}$ is upper bounded by the
probability that $x\not\in \Typ_n(p_l)$ plus the probability that
there exists $x'\neq x$ with $\H^lx' = \H^lx$ $\mod q$. The first
contribution is bounded by Lemma \ref{lemma:TypicalSize}, and we
can therefore neglect it. Denoting the second contribution as
$\block^{(l,*)}$, and writing $\E_x$, $\prob$ for (respectively)
expectation with respect to $x$ and probability with respect to
the matrices $\H_1$, \dots $\H_l$, we have (matrix multiplications
below are understood to be modulo $q$)
\begin{eqnarray*}
\block^{(l,*)} &=&  \E_x\prob\left\{\exists x'\in \Typ_n(p_l)\setminus \{x\}
\mbox{ s.t. } \H^{l}x' = \H^lx\right\}\\
& = & \sum_{t=1}^{l}Q^{(l)}_t\, ,\\
Q^{(l)}_t &\equiv&  \E_x\prob\big\{\exists x'\in \Typ_n(p_l)\setminus \{x\}
\mbox{ s.t. } \\
&&\phantom{\E_x\prob\big\{\exists  x'\in }
\H^{t}x' = \H^tx, \, \H^{t-1}x' \neq \H^{t-1}x\big\}\, .
\end{eqnarray*}
Since, $l\le L=O(\log n)$, it is sufficient to show $Q^{(l)}_t = O(n^{-A})$.
In $Q^{(l)}_t$ we can separate error events due to input $x'$
such that $d_t\equiv
d(\H^{t-1}x',\H^{t-1}x)\le E$ and the other ones.
As a consequence $Q^{(l)}_t$ is upper bounded by
\begin{align*}
\E_x &\prob\left\{\exists x'\in \Typ_n(p_l),
\mbox{ s.t. }
1\le d_t\le E, \; \H^{t}x' = \H^tx\right\}\\
+&\E_x \prob\left\{\exists x'\in \Typ_n(p_l)
\mbox{ s.t. } E< d_t,\, \H^{l}x' = \H^lx\right\}\\
&\le\prob\{\exists z\mbox{ s.t. } ||z||_0\le E\, , \H_tz = 0\}\\
&\phantom{aaaa}+|\T_n(p_l)| \sup\big\{\prob\{\H_t z = 0\}\, :\, ||z||_0>E\big\} .
\end{align*}
Here $z$ is understood to be a $m_{t-1}$ dimensional vector with
entries in $\{0,\dots,q-1\}$.

Notice that $(m_t/km_{t-1})^{1/k}\ge (\delta/k)^{1/k}\ge \delta$.
Next we choose $E =C(q,\Delta=\delta) m_t/\log(m_{t-1}k/m_{t})$ with $C(q,\Delta)$
as in the statement of Lemma \ref{lemma:SmallWeigth}.
As a consequence the first term above is upper bounded by
\begin{eqnarray*}
m_{t-1}^{2}\left(\frac{Bk}{m_t}\right)^{\frac{k}{q}}\le
(Bk)^{\frac{k}{q}} \delta^{-2} m_t^{-{\frac{k}{q}}+2} \le C (\log
n)^{{\frac{k}{q}}-2} n^{-{\frac{k}{q}}+2}\, ,
\end{eqnarray*}
where we used $m_{t-1}\le m_t/\delta$ and $m_t\ge n/(\log n)^2$.
The constant $C$ that depends uniquely on $q$, $k$, $\delta$, but not on $n$.

It remains to bound the second contribution, due to
inputs $x'$ with $d(x',x)> E$. Using
Lemma \ref{lemma:TypicalSize} (to bound $\T_n(p_l)$) and
\ref{lemma:ProbSol} (to bound $\prob\{\H_t z = 0\}$ for $||z||_0>E$)
\begin{eqnarray*}
\E_x \prob\left\{\exists x'\in \Typ_n(p_l)
\mbox{ s.t. } E< d_t,\, \H^{l}x' = \H^lx\right\} \\
\le q^{nh_l+An^{\beta}} (Ckn)^{\frac{q-1}{2}}\left\{\frac{1}{q}
[1+(q-1)e^{-DkE/m_t}]\right\}^{m_t}\, ,
\end{eqnarray*}
By eventually enlarging the constant $A$ (in a way that depends
on $q$), we can
get rid of the term $(Cn)^{\frac{q-1}{2}}$. By
further using $(1+x)\le q^{x/\log q}$ we can upper bound the above by
$k^{q-1/2}q^{\Phi}$ where
\begin{eqnarray*}
\Phi = nh_l+A(q)n^{\beta} -m_t+ A'(q)m_te^{-D(q)kE/m_t}
\end{eqnarray*}
with $A'(q) = (q-1)/\log q$. Notice that $kE/m_t =
C(q,\delta)k/\log(km_{t-1}/m_t)$ can be made arbitrarily large by
taking $k$ large enough. In particular, we can choose
$k_*(q,\delta)$ such that $A'(q)e^{-D(q)kE/m_t}\le \delta/3$ for
any $k\ge k_*$. For such $k$, and using the fact that $m_t\ge m_l
= [nh_l+A(q)n^{\beta}](1+\delta)$
\begin{eqnarray*}
\Phi \le nh_l+A(q)n^{\beta} -m_l(1-\delta/3)\le
 -\frac{1}{3}\delta[nh_l+A(q)n^{\beta}]\, .
\end{eqnarray*}

Summing the various contributions, we obtain, for any $k\ge k_*(q,\delta)$
\begin{eqnarray}
Q^{(l)}_t \le C(q,k,\delta) (\log
n)^{{\frac{k}{q}}-2}n^{-{\frac{k}{q}}+2}+ k^{\frac{q-1}{2}}
q^{-\delta (A(q) n^{\beta}+nh_l)/3}\, ,
\end{eqnarray}
which proves the thesis.
\end{proof}
%
%

\section{Some auxiliary results}
\label{sec:Lemmas}

\begin{proof}[Proof: Lemma \ref{lemma:Entropy}]
First consider Eq.~(\ref{eq:DigitEntropy1}). Let $X_1$ be an
integer random variable with distribution $p$, $X_1(l)$ its $l$-th
least significant $q$-ary digit and $Z$ the indicator function on
$X_1(l)\ge 0$. From $H(X_1(l)) = H(Z)+H(X_1(l)|Z)$ it follows
that, for $\bp_l = \prob\{X_1>q^l\}$:
\begin{eqnarray*}
h_l\le \bp_l\log_q(q-1)-\bp_l\log_q\bp_l -(1-\bp_l)\log_{q}(1-\bp_l)\, .
\end{eqnarray*}
Choosing $q$ large enough so that $\bp_l\le A\,q^{-\eps}\le 1/2$
for all $l\ge 1$, we can upper bound $-(1-\bp_l)\log_{q}(1-\bp_l)$
by $2\bp_l$, thus getting
\begin{eqnarray*}
h_l\le 3\bp_l-\bp_l\log_q\bp_l \, ,
\end{eqnarray*}
which implies  Eq.~(\ref{eq:DigitEntropy1}) for  $\bp_l\le A\,q^{-l\eps}$.)

In order to prove Eq.~(\ref{eq:DigitEntropy2}), first notice that
$H(X_1) = H(\{X_1(l)\})\le \sum_{l\ge 0} H(X_1(l))$ whence
$h_2(p)\le\sum_{l\ge 0} h_l\log_2 q$.
By the same argument $h_2(p)\ge h_0\log_2 q$.
The thesis follows by bounding $\sum_{l\ge 0} h_l$
using Eq.~(\ref{eq:DigitEntropy1}).
\end{proof}

\begin{proof}[Proof: Lemma \ref{lemma:TypicalSize}.]
The number of vectors with type $\theta$ is upper bounded by
$q^{nH(\theta)}$. Since there are at most $(n+1)^q$ distinct
types, $|\Typ_n(p_*)|\le  q^{nH(p_*)+nK_n}$ where
\begin{eqnarray*}
K_n \equiv \sup_{\theta}\{H(\theta)-H(p_*):  D(\theta||p_*)\le n^{-\gamma}\}+\frac{\log_q(n+1)^q}{n}\, .
\end{eqnarray*}
The bound $H(\theta)-H(p_*)\le ||\theta-p_*||_1\log(q/||\theta-p_*||)$
and $||\theta-p_*||\le \sqrt{2 D(\theta||p_*)}$ \cite{CoverThomas}.

Equation (\ref{eq:Sanov}) is just Sanov Theorem.
\end{proof}

\section{Acknoledgments} Yi Lu is supported by the Cisco Stanford Graduate Fellowship.
%
%

\end{document}